\documentclass[11pt]{article}
\usepackage{tikz}
\usetikzlibrary{calc}
 \usepackage{bbm}
 \usepackage[sort]{cite}
\usepackage[sort, numbers]{natbib}
\usepackage{hyperref}
\usepackage{graphicx}
\usepackage{algorithm}
\usepackage[noend]{algorithmic}
\usepackage{epsfig}
\usepackage{epstopdf}
\usepackage{amssymb,amsmath}
\usepackage{amsmath}
\usepackage{amsthm}
\usepackage{amsfonts}
\usepackage{subfigure}
\usepackage{natbib}
\usepackage{psfrag}
\usepackage{rotating}
\usepackage{latexsym}
\usepackage{dsfont}
\usepackage{stmaryrd}
\usepackage{amsthm}
\usepackage{amssymb}
\usepackage{amsmath}
\usepackage[force,almostfull]{textcomp}
\usepackage{fancyvrb}
\usepackage{textcomp}
\usepackage{url}
\usepackage{ifdraft}
\usepackage{multirow}
\usepackage{rotating}
\usepackage{xspace}
\usepackage{array}
\usepackage{flushend}
\usepackage{enumerate}
\usepackage{authblk}
\usepackage[utf8]{inputenc}
\newlength\figureheight
\newlength\figurewidth
\usepackage{graphics}
\usepackage{pgfplots}
\usepackage{siunitx}
\pgfplotsset{compat=newest}
\pgfplotsset{plot coordinates/math parser=false}
\usepackage{scalefnt}

\newtheoremstyle{specialcasestyle}{1mm}{1mm}{\upshape}{}{\bfseries\upshape}{.}{0mm}{}
\theoremstyle{specialcasestyle}

\newtheorem{prop}{Proposition}

\begin{document}

\title{State-dependent Importance Sampling for Estimating Expectations of Functionals of Sums of Independent Random Variables}

\author[1]{Eya Ben Amar}
\author[2]{Nadhir Ben Rached}
\author[3]{Abdul-Lateef Haji-Ali}
\author[1,4]{Ra\'ul Tempone}

\affil[1]{Computer, Electrical and Mathematical Sciences \& Engineering Division (CEMSE), King Abdullah University of Science and Technology (KAUST), Thuwal, Saudi Arabia.}
\affil[2]{Chair of Mathematics for Uncertainty Quantification, Department of Mathematics, RWTH Aachen University, Aachen, Germany.}
\affil[3]{School of Mathematical and Computer Sciences, Heriot-Watt University, Edinburgh, UK.} 
\affil[4]{Alexander von Humboldt Professor in Mathematics for Uncertainty Quantification, RWTH Aachen University, Aachen, Germany.} 

\date{}
\maketitle
\thispagestyle{empty}
\begin{abstract}
Estimating the expectations of functionals applied to sums of random variables (RVs) is a well-known problem encountered in many challenging applications. Generally, closed-form expressions of these quantities are out of reach. A naive Monte Carlo simulation is an alternative approach. However, this method requires numerous samples for rare event problems. Therefore, it is paramount to use variance reduction techniques to develop fast and efficient estimation methods. In this work, we use importance sampling (IS), known for its efficiency in requiring fewer computations to achieve the same accuracy requirements. We propose a state-dependent IS scheme based on a  stochastic optimal control formulation, where the control is dependent on state and time. We aim to calculate rare event quantities that could be written as an expectation of a functional of the sums of independent RVs. The proposed algorithm is generic and can be applied without restrictions on the univariate distributions of RVs or the functional applied to the sum. We apply this approach to the log-normal distribution to compute the left tail and cumulative distribution of the ratio of independent RVs. For each case, we numerically demonstrate that the proposed state-dependent IS algorithm compares favorably to most well-known estimators dealing with similar problems.



\end{abstract}

\section{Introduction}
\subsection{ Motivation}
In a probabilistic model, rare events are important events that infrequently happen with very
small probabilities. Estimating these probabilities has become a substantial area of research because of its many applications, such as queuing systems, insurance risk, financial
engineering, and wireless communication \cite{asmussen2007stochastic,juneja2006rare,glasserman2004monte,rached2015fast}. Typical examples occur in the context of communication systems, where the rare event could be an event in which the system fails to operate properly. For illustration, one can encounter the problem of estimating failure probabilities of the order of $10^{-9}$ for sophisticated networks, such as ultra-reliable fifth or sixth generation (5G or 6G) systems \cite{rached2020accurate}.

Calculating rare event quantities that could be written as an expectation of a functional of the sums of independent RVs is of paramount practical interest in many challenging applications. For instance, in financial engineering, calculating the value-at-risk (VaR) requires computing the left tails of the sums of RVs (i.e., the probability that the sum is less than a sufficiently small threshold).
Another relevant example is calculating the probability that the signal-to-interference-plus-noise ratio is less than a given threshold in communication systems. Under some particular fading environments, this probability can be expressed as a cumulative distribution function (CDF) of the ratio of independent RVs.

 \subsection{ Literature Review}
Various researchers have proposed closed-form approximations of the left and right tails
of sums of RVs \cite{lopez2009simple,xiao2019outage,chatterjee2018downlink,beaulieu2020marcum,zhu2020right,constantinescu2016closed,singh2018exact}. However, these approximation methods are not generic. Moreover, their accuracy is not always guaranteed for all scenarios, as it can degrade for certain system parameters. The Monte Carlo (MC) method can be used as a generic tool to cope with these problems. However, it is well acknowledged that estimating rare event quantities with the naive MC sampler requires a prohibitively large number of simulation runs \cite{kroese2013handbook}. Variance reduction techniques have been used extensively to improve the computational work of the naive MC method. In this context, importance sampling (IS) is among the most popular variance reduction techniques that provide accurate estimates of rare event probabilities with a reduced number of simulation runs when appropriately used \cite{kroese2013handbook}.

Variance reduction techniques have been widely discussed in the literature, and particular focus has been devoted to proposing algorithms for the efficient simulation of the right tail of sums of RVs (i.e., the probability that the sum exceeds a sufficiently large threshold). In particular, for distributions with light right tails (i.e., decaying at an exponential rate or faster), under some regularity assumptions, the popular exponential twisting IS approach \cite{asmussen2007stochastic} satisfies the logarithmic efficiency property, which is a useful metric used to assess the efficiency of an estimator. In contrast, for heavy-tailed distributions, such as the case of log normals and Weibulls with shape parameters strictly less than 1, the exponential twisting method is inapplicable. Therefore, efficient algorithms have been developed for estimating tail probabilities involving heavy-tailed RVs. In this context, \cite{asmussen1997simulation} provided the first logarithmically efficient estimator for such probabilities using the conditional MC idea. Other authors \cite{asmussen2006improved} have proposed an estimator with a bounded relative error under distributions with regularly varying tails, which was further extended to more general scenarios (e.g., see \cite{hartinger2009efficiency,chan2011rare,asmussen2012error,asmussen2015error,ghamami2012improving}). In addition to estimators based on the conditional MC, various state-independent IS techniques have been proposed \cite{juneja2002simulating,juneja2007estimating,am2012state,murthy2015state}. 

State-independent changes of measure for estimating certain rare events involving sums of heavy-tailed RVs cannot achieve logarithmic efficiency \cite{bassamboo2007inefficiency}. 
Therefore, more complex state-dependent IS algorithms have been proposed in the literature over the last few years to estimate probabilities for sums of heavy-tailed independent RVs. Of valuable interest are studies developed in \cite{dupuis2007importance,dupuis2007subsolutions,blanchet2006efficient,blanchet2008state,blanchet2011rare,blanchet2012state}. The researchers in \cite{blanchet2006efficient} developed an efficient state-dependent IS estimator with a bounded relative error under distributions with regularly varying heavy tails. The estimator can also be adapted to provide strongly efficient algorithms in light-tailed situations. A related approach, based on the construction of Lyapunov inequalities, has also been developed \cite{blanchet2008state} for constructing strongly efficient estimators for large deviation probabilities of regularly varying random walks. These algorithms use a parametric family of change of measure based on mixtures that are appropriately selected using Lyapunov bounds. 
Moreover, stochastic control and game theory have been used to build efficient state-dependent IS schemes to simulate rare events \cite{dupuis2004importance,dupuis2005dynamic,dupuis2007importance}. For instance, in the heavy-tailed setting, the authors in \cite{dupuis2007importance} constructed dynamic IS estimators with a nearly asymptotically optimal relative error for independent and identically distributed (i.i.d.) nonnegative regularly varying RVs. They considered a parametric family of change of measure whose parameters are determined by solving a deterministic, discrete-time control problem. The closest work to the proposed approach is in \cite{dupuis2004importance}, where the authors proposed an approach based on connecting IS with stochastic optimal control (SOC). The scope of \cite{dupuis2004importance} is limited to the i.i.d. case and distributions with finite moment generating functions. In this work, independence is the only assumption we make. The connection between IS and SOC has been investigated for other scenarios, such as diffusions \cite{hartmann2017variational}, and for stochastic reaction networks approximated by the Tau-Leap scheme \cite{hammouda2021optimal}. The dynamics in \cite{hammouda2021optimal} evolve according to discrete-time discrete-space Markov chains, whereas the dynamics in this work evolve according to discrete-time continuous-space Markov chains.

Few researchers have recently addressed the left-tail region (i.e., the probability that sums of nonnegative RVs fall below a sufficiently small threshold \cite{asmussen2016exponential,issaid2017efficient,rached2015unified, rached2020accurate,rached2020universal,large_sum}). For instance, \cite{asmussen2016exponential} considered the specific setting of the i.i.d. sum of log-normal RVs. The approach was based on the exponential twisting technique and is logarithmically efficient. 
The work of \cite{rached2015unified} proposed two unified hazard rate twisting (HRT)-based approaches that estimate the outage capacity values for generalized independent fading channels. The first estimator achieves logarithmic efficiency for arbitrary fading models, whereas the second achieves the bounded relative error criterion for most well-known fading variates and logarithmic efficiency for the log-normal case. Recently, \cite{rached2020accurate} proposed an IS scheme based on sample rejection applied to the case of the independent Rayleigh, correlated Rayleigh, and i.i.d. Rice fading models, showing that the estimator satisfies the bounded relative error property. 

\subsection{Main Contributions}
In this paper, we propose a generic state-dependent IS approach to estimate rare event probabilities that could be written as an expectation of a functional of the sums of independent RVs. We adopt a SOC formulation to determine the optimal IS parameters, minimizing the variance or, equivalently, the second moment of the estimator within a preselected class of measures. After formulating the SOC problem and describing the algorithm used to derive the optimal controls, which are optimal IS parameters, we apply the proposed algorithm to two examples: the computation of the left-tail probability in a log-normal setting, and the computation of the CDF of the ratio of independent log-normal RVs. The proposed algorithm is generic and not restricted to the log-normal environment. The algorithm can be applied to compute the quantity of interest without restrictions on the distribution of the univariate RVs in the sum or the expression of the functional applied to the sum. Numerical simulations demonstrate the superior performance of the proposed estimator in terms of the number of samples and computational work to meet a given prescribed tolerance (TOL) compared to the existing state-of-the-art estimators dealing with similar problems.

The rest of the paper is organized as follows. Section~\ref{Section 1} describes the problem setting, presents applications that fall within the scope of applicability of the proposed approach, and introduces the concept of IS. Section~\ref{Section 2} contains the main work, explaining the state-dependent IS scheme via a novel SOC formulation, followed by presenting the algorithm. Section~\ref{Section 3} applies the proposed algorithm to two applications in wireless communications. The proposed algorithm compares favorably to some well-known estimators dealing with similar problems. 

\section{Problem Setting}\label{Section 1}
This section states the objective of the method and enumerates some applications that fall within the scope of its applicability. Next, this section introduces the concept of the naive MC method. Finally, it presents the IS technique, one of the most popular variance reduction techniques.\\

\subsection{Objective}
We consider $\mathbf{X}=(X_1,X_2,\cdots, X_N)^t$ to be a random vector comprising independent positive components with probability density functions (PDFs) $f_{X_{1}}(.), f_{X_{2}}(.), \ldots, f_{X_{N}}(.)$ and a joint PDF $f(\mathbf{x})=\prod_{n=1}^{N} f_{X_{n}}\left(x_{n}\right)$.  In this work, $X_i, i=1, \cdots, N$ are one-dimensional vectors. However, this approach is still applicable to the multidimensional case. We let $S_N=\sum_{n=1}^{N} X_i$ and $g: \mathbb{R}_{+} \rightarrow \mathbb{R}$ be a given function.
We aim to develop a state-dependent IS algorithm via a connection to an SOC formulation to  estimate rare event quantities that could be written as follows:
 \begin{equation}
     \label{eqn:eq1}
\alpha=\mathbb{E} \left[g\left(S_N \right)\right].
 \end{equation}
\subsection{Applications}
\subsubsection{Right and left tail}
One of the problems that can be written as (\ref{eqn:eq1}) is calculating the right-tail probability of sums of RVs (i.e., the probability that the sum is larger than a sufficiently large threshold), which arises in many areas of engineering. This probability can be expressed as 
\begin{equation}
\label{right}
    \alpha=\mathbb{P}(S_{N} \geq \gamma_{\textrm{th}})=\mathbb{E} \left[\mathbbm{1}_{(S_{N} \geq \gamma_{\textrm{th}})} \right],
\end{equation}
corresponding to (\ref{eqn:eq1}), where $g(x)=\mathbbm{1}_{(x \geq \gamma_{\textrm{th}})}$. \\
As a practical example, the right-tail probability of the sums of RVs may represent the ruin probability of an insurance company. In this setting, $S_N$ represents the total sum of claims, and $\gamma_{\textrm{th}}$ is the initial reserve. The claims $X_1,\cdots,X_N$ can be modeled by heavy-tailed distributions \cite{asmussen2000rare}. In the Cramer--Lundberg model, this probability can be expressed as (\ref{right}) \cite{asmussen2007stochastic}. 
Moreover, calculating left-tail probabilities occurs extensively in many applications. In these cases, the quantity of interest can be expressed as 
\begin{equation}
\label{left}
    \alpha=\mathbb{E} \left[\mathbbm{1}_{(S_N \leq \gamma_{\textrm{th}})} \right],
\end{equation}
which is in the form of (\ref{eqn:eq1}), where $g(x)=\mathbbm{1}_{(x \leq \gamma_{\textrm{th}})}$. \\

One of the most relevant examples is estimating the VaR, defined as the $1-\alpha$ quantile of the loss distribution, for a sufficiently small value of $\alpha$. 
We let a portfolio be based on $N$ assets with upcoming prices $X_1\cdots, X_N$, which can be modeled using log-normal distributions \cite{asmussen2016exponential}. The VaR of $S_{N}$ at the level of $\alpha$ $(\operatorname{VaR}_{\alpha}\left(S_{N}\right))$ is defined as the value such that the probability of a loss larger than that value is equal to $1-\alpha$ \cite{alemany2013nonparametric,sun2009general}. In other words, $\operatorname{VaR}_{\alpha}\left(S_{N}\right)=F_{S_N}^{-1}(\alpha)$, where $F_{S_N}(.)$ is the CDF of $S_N$. \\

Another challenging application is the analysis of wireless communication systems. The outage probability (OP), defined as the probability that the signal-to-noise ratio (SNR) falls below a given threshold $\gamma_{\textrm{th}}$ \cite{yilmaz2012unified}, is equivalent to computing the CDF of sums of the SNRs of the received signals. Hence, it can be expressed as in (\ref{left}). \\

\subsubsection{CDF of the ratio of independent RVs} 
Another performance metric that can be expressed as in (\ref{eqn:eq1}) is the OP in the presence of co-channel interference and noise. 
For single-input, single-output (SISO) systems, the OP is expressed as \cite{rached2017efficient}
  \begin{equation}
\label{eqn:OPsinr}
P_{out}=\mathbb{P}\left( \operatorname{SINR} \leq \gamma_{\textrm{th}} \right)=\mathbb{P}\left(\frac{X_{0}}{\sum_{n=1}^{N} X_{n}+ \eta} \leq \gamma_{\textrm{th}}\right),
\end{equation}
where $X_0$ denotes the desired user signal power, $X_1, \cdots, X_N$ represent the received powers of the $N$ interfering signals, and $\eta$ indicates the variance of the additive white Gaussian noise. We assume that $X_0,\cdots, X_N$ are independent. Through conditioning on $X_{1}, X_{2}, \cdots, X_{N}$ and using the law of total expectation, we write (\ref{eqn:OPsinr}) as
  \begin{equation}
\label{eqn:OPCDF}
 \mathbb{E}\left[F_{X_{0}}\left(\gamma_{\textrm{th}} \sum_{n=1}^{N} X_{n}+ \gamma_{\textrm{th}} \eta \right)\right],
\end{equation}
where $F_{X_{0}}(\cdot)$ is the CDF of the RV $X_{0}$, corresponding to the form in (\ref{eqn:eq1}) with $ g(x)=F_{X_{0}}(\gamma_{\textrm{th}}(x+\eta))$.

\subsection{Importance Sampling}
The naive MC estimator of the quantity of interest in (\ref{eqn:eq1}) is 
\begin{equation}
\hat{\alpha}_{m c}=\frac{1}{M} \sum_{k=1}^{M} g\left(S_N^{(k)} \right),
\end{equation}
where $M$ denotes the number of simulation runs, and
$\{S_{N}^{(k)}\}_{k=1}^{M}$ represents independent realizations of the {RV} $S_{N}=\sum_{i=1}^{N} X_{i}$. 
\\

However, the naive MC method is computationally expensive, requiring a substantial number of simulation runs to meet a given accuracy when considering rare event probabilities. Using appropriate variance reduction techniques, such as IS, is necessary to overcome the failure of naive MC simulations and considerably reduce the computational work. The idea is to perform a change of measure under which the rare event is generated with a higher probability than under the original distribution \cite{kroese2013handbook}. The IS technique consists of writing $\alpha$ as
\begin{equation}
\alpha =\mathbb{E}_{\tilde{f}}\left[\tilde{g}\left(\mathbf{X}\right)\right],
\end {equation}
where 
\begin{equation}
\tilde{g}\left(\mathbf{x}\right)=g\left(\sum_{n=1}^{N} x_{n}\right) \frac{f\left(\mathbf{x}\right)}{\tilde{f}\left(\mathbf{x}\right)},
\end{equation}
and $\mathbb{E}_{\tilde{f}}[\cdot]$ denotes the expectation under which the vector $\mathbf{X}$ has the joint PDF $\tilde{f}(\cdot)$. The IS estimator is expressed as
\begin{equation}
\hat{\alpha}_{I S}=\frac{1}{M} \sum_{k=1}^{M} \tilde{g}\left(\mathbf{X}^{(k)} \right),
\end{equation}
where $ \{ \mathbf{X}^{(k)}  \}_{k=1}^{M}$ represents independent realizations of $\mathbf{X}$ sampled according to $\tilde{f}(\cdot)$. When $g(x)>0, x\in \mathbb{R}_+$, the optimal change of measure minimizing the variance of the IS estimator is given by
\begin{equation}
\begin{aligned}
    \label{eqn:eqocom}
f^{*}(\mathbf{x} )=\frac{f(\mathbf{x} )  g \left(\sum_{n=1}^{N} x_{i}\right)}{\alpha}, \; \mathbf{x} \in \mathbb{R}_+^N.
\end{aligned}
\end {equation}
This optimal change of measure yields zero variance; thus, it is called the zero variance change of measure. However, using such a change of measure is impractical because it assumes the knowledge of $\alpha$, which is the unknown quantity. 

\section{IS via an SOC Formulation}\label{Section 2}
This section explains the SOC formulation and how to link it to IS to construct the state-dependent IS estimator. Then, it introduces the HRT family as a change of measure. Finally, this section describes the steps of the proposed state-dependent IS algorithm.

\subsection{State-dependent IS Approach}
The idea we adopt is to link the problem of finding an efficient change of measure to a SOC problem. To apply SOC to the current static problem, we embed it with the evolution of a Markov chain with the following dynamics:
\begin{equation}
\label{eqn:eqd1}
S_{n+1}=S_{n}+X_{n+1},\; \;  n=0,1, \cdots, N-1,
\end{equation}
where $S_{0}=0$. Instead of sampling $X_{n+1}$ according to $f_{X_{n+1}}(\cdot)$, we perform a change of measure such that, given $S_{n}, \; X_{n+1}$ is distributed according to $\Tilde{f}_{X_{n+1}}\left(\cdot ; \mu_{n+1}(S_n)\right)$, where $\mu_{n+1}$ is a function of $S_{n}$. With this idea, the new joint PDF can be written as 
\begin{equation}
\Tilde{f}\left(\mathbf{x}\right)=\prod_{n=1}^{N} \Tilde{f}_{X_{n}}\left(x_{n} ; \mu_{n}\left( s_{n-1}\right)\right),
\end{equation}
where $s_{n-1}=\sum_{i=1}^{n-1} x_i$. The objective is to determine the optimal controls $\mu_{n}: \mathbb{R}_+ \rightarrow A \subset \mathbb{R}, \;  n=1,2, \cdots, N$ that minimize the second moment of the IS estimator. Therefore, assuming that the second moment of the estimator is finite, we define the cost function for $\mu_{n+1}, \cdots, \mu_{N} \in \mathcal{D}^{N-n},\; n=0, \cdots, N-1$ as 
\begin{equation}
\begin{aligned}
C_{n, s}(\mu_{n+1}, \cdots, \mu_{N})  =\mathbb{E}_{\Tilde{f}}\left[\left(g\left(S_{N}\right)\right)^{2} \prod_{i=n+1}^{N}\left(\frac{f_{ X_{i}}\left(X_{i}\right)}{\Tilde{f}_{X_{i}}\left(X_{i} ; \mu_{i}(S_{i-1})\right)}\right)^{2} \mid
S_{n}=s\right],
\end{aligned}
\end{equation}
where $\mathcal{D}=\{ \mu : \mathbb{R^+} \to A \}$ represents the set of admissible Markov controls. More precisely, $$\text{for} \quad \mu \in \mathcal{D}, \quad \quad \mu: \mathbb{R}_+ \rightarrow A.$$
We also define the value function as follows:
\begin{equation}
u(n, s)=\inf _{(\mu_{n+1}, \cdots, \mu_{N})  \in \mathcal{D}}^{N-n} C_{n, s}(\mu_{n+1}, \cdots, \mu_{N}).
\end{equation}
The above SOC formulation is flexible because the RVs are dependent. The same observation holds for the optimal change of measure (\ref{eqn:eqocom}). Therefore, if the family of PDFs $\Tilde{f}_{X_{n}}\left(.; \mu_{n}\right)$ is sufficiently large, we can expect the SOC formulation to deliver an estimator with a performance close to that of the optimal estimator.

Next, the question is how to solve the minimization problem and determine the optimal controls $\mu_{n}, \; n=1,2, \cdots, N$. The idea is to solve it sequentially by going backward in time. In Proposition \ref{prop}, we state the dynamic programming equation solved by the value function $u$. 
\begin{prop}
\label{prop}
For all $n \in\{0,1, \cdots, N-1\}$ and $s \geq 0$, we obtain
\begin{equation}
\begin{aligned}
\label{eqn:eq0}
&u(n, s) =\inf _{\mu \in A} \mathbb{E}_{\Tilde{f}}\left[\left(\frac{f_{ X_{n+1}}\left(X_{n+1}\right)}{\Tilde{f}_{X_{n+1}}\left(X_{n+1} ; \mu\right)}\right)^{2} u\left(n+1, S_{n+1}\right) \mid S_{n}=s\right].
\end{aligned}
\end{equation}
If the minimum is attained, we have \\
\begin{equation}
\begin{aligned}
\label{eqn:eq00}
\mu_{n+1}(s)=
\arg \min _{\mu \in A} \; \mathbb{E}_{\Tilde{f}}\left[\left(\frac{f_{ X_{n+1}}\left(X_{n+1}\right)}{\Tilde{f}_{X_{n+1}}\left(X_{n+1} ; \mu\right)}\right)^{2} u\left(n+1, S_{n+1}\right) \mid S_{n}=s\right],
\end{aligned}
\end{equation}
where $u(N, x)=(g(x))^{2}$,  $S_{n+1}= s +  X_{n+1}$, and $X_{n+1}$ is distributed according to $\Tilde{f}_{X_{n+1}}\left(\cdot; \mu_{n+1}(s)\right)$.
\\ \\
\end{prop}
\begin{proof}
For simplicity, we assume that the optimal control is attained:
\begin{equation}
	\label{min}
	u(n, s)=\min _{(\mu_{n+1}, \cdots, \mu_{N})  \in \mathcal{D}^{N-n}} C_{n, s}(\mu_{n+1}, \cdots, \mu_{N}).
\end{equation}
\textbf{Step 1}
We let $\mu_{n+1}^*, \cdots, \mu_{N}^*$ be the optimal control minimizing (\ref{min}). Then, we obtain 
\begin{equation}
	\begin{aligned}
		\label{eqn:eqproof}
		 u(n, s)&= \mathbb{E}_{\Tilde{f}}\left[\left(g\left(S_{N}\right)\right)^{2} \prod_{i=n+1}^{N}\left(\frac{f_{ X_{i}}\left(X_{i}\right)}{\Tilde{f}_{X_{i}}\left(X_{i} ; \mu_{i}^* (S_{i-1})\right)}\right)^{2} \mid S_{n}=s\right]
		\\&= \mathbb{E}_{\Tilde{f}}\left[\mathbb{E}_{\Tilde{f}}\left[\left(g\left(S_{N}\right)\right)^{2}  \prod_{i=n+2}^{N}\left(\frac{f_{ X_{i}}\left(X_{i}\right)}{\Tilde{f}_{X_{i}}\left(X_{i} ; \mu_{i}^*(S_{i-1})\right)}\right)^{2} \right. \right. \\&   \times \left. \left. \left(\frac{f_{ X_{n+1}}\left(X_{n+1}\right)}{\Tilde{f}_{X_{n+1}}\left(X_{n+1} ; \mu_{n+1}^*(S_{n})\right)}\right)^{2} \mid S_{n}=s,X_{n+1}\right] \mid S_{n}=s\right].
	\end{aligned}
\end{equation}
Knowing $X_{n+1}$ and $S_n$, $\left(\frac{f_{ X_{n+1}}\left(X_{n+1}\right)}{\Tilde{f}_{X_{n+1}}\left(X_{n+1} ; \mu_{n+1}^*(S_{n})\right)}\right)^{2}$ is deterministic. Thus, using the Markov property of $S_n$, we obtain
\begin{equation}
	\begin{aligned}
		 \mathbb{E}_{\Tilde{f}}\left[\left(g\left(S_{N}\right)\right)^{2} \prod_{i=n+2}^{N}\left(\frac{f_{ X_{i}}\left(X_{i}\right)}{\Tilde{f}_{X_{i}}\left(X_{i} ; \mu_{i}^*(S_{i-1})\right)}\right)^{2} \mid S_{n}=s,X_{n+1}\right]=&C_{n+1,S_{n+1}}(\mu_{n+2}^*, \cdots, \mu_{N}^*) \\ &\geq u(n+1,S_{n+1}).
	\end{aligned}
\end{equation}
Hence, the following inequality holds:
\begin{equation}
	\begin{aligned}
		\label{eqn:eqproof2}
		& u(n, s) \geq \mathbb{E}_{\Tilde{f}}\left[\left(\frac{f_{ X_{n+1}}\left(X_{n+1}\right)}{\Tilde{f}_{X_{n+1}}\left(X_{n+1} ; \mu_{n+1}^*(s)\right)}\right)^{2} u(n+1,S_{n+1}) \mid S_{n}=s\right]
		\\ & \geq \min _{\mu\in {A}} \mathbb{E}_{\Tilde{f}}\left[\left(\frac{f_{ X_{n+1}}\left(X_{n+1}\right)}{\Tilde{f}_{X_{n+1}}\left(X_{n+1} ; \mu\right)}\right)^{2} u\left(n+1, S_{n+1}\right) \mid S_{n}=s\right].
	\end{aligned}
\end{equation}
\textbf{Step 2}
We choose the control $\mu_{n+1}^+$ to be arbitrary and, given the value of $S_{n+1}$, we select the optimal controls $\mu_{n+2}^*, \cdots, \mu_{N}^*$. Then, the following lower bound holds:
\begin{equation}
	\begin{aligned}
		\label{eqn:eqproof3}
		& u(n, s)  \leq \mathbb{E}_{\Tilde{f}}\left[\left(\frac{f_{ X_{n+1}}\left(X_{n+1}\right)}{\Tilde{f}_{X_{n+1}}\left(X_{n+1} ; \mu_{n+1}^+(s)\right)}\right)^{2} \left(g\left(S_{N}\right)\right)^{2}  \prod_{i=n+2}^{N}\left(\frac{f_{ X_{i}}\left(X_{i}\right)}{\Tilde{f}_{X_{i}}\left(X_{i} ; \mu_{i}^*(S_{i-1})\right)}\right)^{2} \mid S_{n}=s\right]
		\\&  \leq \mathbb{E}_{\Tilde{f}}\left[\left(\frac{f_{ X_{n+1}}\left(X_{n+1}\right)}{\Tilde{f}_{X_{n+1}}\left(X_{n+1} ; \mu_{n+1}^+(s)\right)}\right)^{2}  \right.  \\ & \;  \; \; \; \; \;  \times  \mathbb{E}_{\Tilde{f}}\left[\left(g\left(S_{N}\right)\right)^{2}  \left. \prod_{i=n+2}^{N}\left(\frac{f_{ X_{i}}\left(X_{i}\right)}{\Tilde{f}_{X_{i}}\left(X_{i} ; \mu_{i}^*(S_{i-1})\right)}\right)^{2} \mid S_{n}=s,X_{n+1}\right] \mid S_{n}=s\right]
		\\& = \mathbb{E}_{\Tilde{f}}\left[\left(\frac{f_{ X_{n+1}}\left(X_{n+1}\right)}{\Tilde{f}_{X_{n+1}}\left(X_{n+1} ; \mu_{n+1}^+ (s)\right)}\right)^{2}u(n+1,S_{n+1}) \mid S_{n}=s\right].
	\end{aligned}
\end{equation}
Taking the minimum over all controls $\mu_{n+1}^+(s)$ yields
\begin{equation}
	\label{second_inq}
	\begin{aligned}
		 u(n, s) \leq \min _{\mu \in {A}} \mathbb{E}_{\Tilde{f}}\left[\left(\frac{f_{ X_{n+1}}\left(X_{n+1}\right)}{\Tilde{f}_{X_{n+1}}\left(X_{n+1} ; \mu\right)}\right)^{2} u\left(n+1, S_{n+1}\right) \mid S_{n}=s\right].
	\end{aligned}
\end{equation}
Hence, the proof is concluded using (\ref{eqn:eqproof3}) and (\ref{second_inq}).
\end{proof}
\textit{Remark:}
We can prove the proposition without the assumption that the minimum is attained. For that, we use a minimizing sequence $\mu_i$ of controls, satisfying
\begin{equation}
	u(n,s)=\lim _{i \to \infty} C_{n,s}(\mu_i).
\end{equation}

\subsection{Hazard Rate Twisting Family}
The choice the family of PDFs $\Tilde{f}_{X_{n}}\left(.; \mu_{n}\right), n=1,\cdots,N$ in this work is based on the well-known HRT. The HRT technique was originally developed to deal with the right tail of sums of heavy-tailed RVs \cite{juneja2002simulating,rached2018generalization}.\\

We define the hazard rate $\lambda_{X_i}(\cdot)$ associated with the RV $X_i$ as
 \begin{equation}
     \label{eqn:eq4}
\lambda _{X_i}(x)=\frac{f_{X_i}(x)}{1-F_{X_i}(x)}, \; \; x>0,
\end{equation}
where $F_{X_i}(x)=\mathbb{P}(X_i \leq x)$ is the CDF of $X_i$, $i=1,\cdots, N$. We also define the hazard function as
 \begin{equation}
     \label{eqn:eq5}
\Lambda_{X_i}(x)=-\log \left(1-F_{X_i}(x)\right), \; \; x>0.
\end{equation}
From (\ref{eqn:eq4}) and (\ref{eqn:eq5}), the PDF of $X_{i}$ can be expressed as 
 \begin{equation}
     \label{eqn:eq6}
f_{X_i}(x)=\lambda_{X_i}(x) \exp \left(-\Lambda_{X_i}(x)\right), \; \; x>0.
\end{equation}
The HRT change of measure is obtained by twisting the hazard rate of each component $X_i \; , \; i=1,\cdots, N$ by a quantity $\mu_i<1$ as follows:
 \begin{equation}
     \label{eqn:HRT}
\begin{aligned}
\Tilde{f}_{X_i}(x;\mu_i) & = (1-\mu_i) \lambda_{X_i}(x) \exp \left(-(1-\mu_i) \Lambda_{X_i}(x)\right) \\
&=(1-\mu_i) f_{X_i}(x) \exp \left(\mu_i \Lambda_{X_i}(x)\right), \; \; x>0.
\end{aligned}
\end{equation}
\\
 Moreover, $\mu_i$ should satisfy $0\leq\mu_i<1, \; i=1,\cdots, N$ to efficiently address the estimation of the right tail of the sum distribution. Consequently, the tail of the resulting distribution becomes much heavier to the right than the original. However, this feature is unsuitable for dealing with the left tail. Two approaches were proposed in \cite{rached2015unified} to adjust the HRT to handle the left-tail region. The first is based on twisting the RVs $-X_1,\cdots,-X_N$ instead of the original variates $X_1,\cdots, X_N$. The second approach applies the HRT approach to $X_1,\cdots, X_N$ using a negative twisting parameter. \\ 
 Considering the appropriate twisting parameter, we employ the HRT change of measure given by (\ref{eqn:HRT}), and the set $A$ in this case is given by ${A}=(-\infty,1)$. By doing so, the value function is given by
\begin{equation}
\begin{aligned}
&u(n, s) =\inf _{\mu \in A} \mathbb{E}_{\Tilde{f}}\left[\frac{\exp \left(-2\mu \Lambda_{X_{n+1}}(X_{n+1})\right)}{(1-\mu)^{2} } u\left(n+1, S_{n+1}\right) \mid S_{n}=s\right].
\end{aligned}
\end{equation}

\subsection{Algorithm}
Based on the results stated in the proposition, we propose a numerical algorithm to approximate the optimal controls $\mu_n$, where $n=1,\cdots, N$. We start by truncating the space $\mathbb{R}_{+}$ and work in the interval $[0, \Bar{S}]$, where $\Bar{S}$ is a large number in $\mathbb{R}_{+}$. There are particular cases that we treat, where $\Bar{S}$ is naturally chosen. For instance, when estimating $\mathbb{P}(S_N \leq \gamma_{\textrm{th}})$, due to the nonnegativity of $X_i$, $u(n,s)=0$ for $s \geq \gamma_{\textrm{th}}$ and $n=0, \cdots, N$. In this case, $\Bar{S}$ is set equal to $\gamma_{\textrm{th}}$. In the general case, $\Bar{S}$ is selected to be sufficiently large. At each step of the backward algorithm, we use linear extrapolation to compute the value function for $s > \Bar{S} $.

We consider a mesh in the one-dimensional $s$-space: $0=s_{0}<s_{1}, \cdots<s_{K}=\Bar{S}$. The aim is to approximately compute $u\left(n, s_{k}\right)$ for all $n=0,1, \cdots, N-1$ and $s_{k}, \, k=0,1, \cdots, K$. The algorithm is summarized as follows:\\\\
\textbf{Step 1:} For each $s_{k}$ in the mesh, we solve the following:
\begin{equation}
\label{min_prob}
\begin{aligned}
 u\left(N-1, s_{k}\right)  & =\min _{\mu \in A} \mathbb{E}_{\Tilde{f}}\left[\left(\frac{f_{X_{N}}\left(X_{N}\right)}{\Tilde{f}_{X_{N}}\left(X_{N} ; \mu\right)}\right)^{2}\left(g\left(s_{k}+X_{N}\right)\right)^{2}\right] \\
&=\min _{\mu\in A} \int_{0}^{+\infty} \frac{\left(f_{X_{N}}(t)\right)^{2}}{\Tilde{f}_{X_{N}}\left(t ; \mu\right)}\left(g\left(s_{k}+t\right)\right)^{2} dt,
\end{aligned}
\end{equation}
and
\begin{equation}
\label{min_prob1}
\begin{aligned}
& \mu_N(s_k) 
&=\underset{\mu\in A}{\arg\min} \int_{0}^{+\infty} \frac{\left(f_{X_{N}}(t)\right)^{2}}{\Tilde{f}_{X_{N}}\left(t ; \mu\right)}\left(g\left(s_{k}+t\right)\right)^{2} dt.
\end{aligned}
\end{equation}
This step is not expensive because we must compute a one-dimensional integral for each point in the mesh and perform an optimization problem for the parameter $\mu$. When the HRT family is used, the optimization problem becomes equivalent to determining the root of a nonlinear equation.\\\\
\textbf{Step 2:} After obtaining $u\left(N-1, s_{k}\right)$ for all $s_{k}$ in the grid, the next step again applies the result of the proposition to obtain an approximation of $u(N-2, s_{k})$ and $ \mu_{N-1}(s_k) $
\begin{equation}
\begin{aligned}
 u\left(N-2, s_{k}\right) =\min _{\mu \in A} \int_{0}^{+\infty} \frac{\left(f_{X_{N-1}}(t)\right)^{2}}{\Tilde{f}_{X_{N-1}}\left(t ; \mu \right)}u\left(N-1,s_{k}+t\right) dt.
\end{aligned}
\end{equation}
To perform this step, we must know $u(N-1, s)$ for all $s$ that are not necessarily in the grid. To overcome this problem, we proceed by interpolating between the points $u\left(N-1, s_{k}\right)$, where $k=0,1,\cdots,K$. As mentioned, linear extrapolation is employed for $s>\Bar{S}$ when needed.\\\\
\textbf{Step 3:} After computing $ \mu_{n}(s_k) $ for $n=1,2, \cdots, N$ and all $s_{k}$ in the grid $k=0,1,2, \cdots, K$, the following step is to solve for $\mu_{n}, n=1,2, \cdots, N$ by going forward in time. More specifically, we start at $S_{0}=0$ and sample from $\Tilde{f}_{ X_{1}}\left(\cdot, \mu_{1}\right)$ to obtain $S_{1}$. Further, $\mu_{1}(0)$ was already computed in the resolution of the backward problem. We compute $\mu_{2}$ as
\begin{equation}
\label{eqn:eq11}
\mu_2\left(\Tilde{s}_{1}\right)= \underset{\mu \in A}{\arg\min} \int_{0}^{\infty} \frac{\left(f_{X_{2}}(t)\right)^{2}}{\Tilde{f}_{X_{2}}\left(t ; \mu\right)} u\left(2, \Tilde{s}_{1}+t\right) dt.
\end{equation}
After computing $\mu_2$, we simulate $S_{2}$ as $S_{2}=\Tilde{s}_{1}+X_{2}$, with $X_2$ sampled from $\Tilde{f}_{X_{2}}\left(. ; \mu_{2}\right)$. We continue repeating this procedure until we reach $\mu_N $ and then sample $X_{N}$.
In the case of smooth controls, the optimization problem (\ref{eqn:eq11}) can be avoided using interpolation between controls, obtained in the backward step, on the grid, $s_1,\cdots, s_K$. \\\\
\textbf{Step 4:} The forward problem is repeated $M$ times. The proposed IS estimator is given as
\begin{equation}
\hat{\alpha}_{\textrm{IS}}=\frac{1}{M} \sum_{k=1}^{M} g\left(S_{N}^{(k)}\right) \prod_{i=1}^{N} \frac{f_{X_{i}}\left(X_{i}^{(k)}\right)}{\Tilde{f}_{X_{i}}\left(X_{i}^{(k)}, \mu_i(S_{i-1}^{(k)})\right)}.
\end{equation}
 \section{Numerical Results}\label{Section 3}
 This section presents selected numerical results to illustrate the performance of the proposed IS scheme. First, the methodology adopted to demonstrate the performance of the proposed approach is discussed. The motivation for using the improved version of the proposed method, called the aggregate method, is explained. Then, the proposed algorithm is applied to estimate the OP at the output of diversity receivers with and without co-channel interference in the log-normal environment.

\subsection{Methodology}
Within the broad applicability of the proposed estimator, we focused on applying it to calculate the left-tail probability and the CDF of the ratio of independent RVs. We used the proposed estimator to estimate the OP at the output of diversity receivers with and without co-channel interference. We considered the case in which the antennae are sufficiently spaced to assume that independent RVs can model fading channels. We considered the log-normal fading environment that exhibits a good fit for realistic propagation channels. We demonstrated that the proposed approach achieves a substantial reduction of the variance compared to other well-known IS algorithms. \\
In both applications, the objective was to efficiently estimate the following:
 \begin{equation}
     \label{eqn:eq1.}
\alpha=\mathbb{E} \left[g\left(\sum_{i=1}^{N} X_i \right)\right],
 \end{equation}
where $X_1,\cdots, X_N$ denote i.i.d. log-normal RVs with parameters $m$ and $\sigma^2$. The PDF of $X_i, \; i=1,\cdots, N,$ is expressed as follows:
 \begin{equation}
 f_{X_i}(x)=\frac{1}{x \sigma \sqrt{2 \pi}} \exp \left(-\frac{(\ln x-m)^{2}}{2 \sigma^{2}}\right) , \; \; x >0.
 \end{equation}
 For the second application, we let $X_0$ be a log-normal RV with parameters $m_0$ and $\sigma_0^2$.
 
We employed the HRT change of measure in (\ref{eqn:HRT}) to build the estimator. Hence, we call this approach the HRT-SOC IS approach, and the corresponding estimator is denoted by $T_{\textrm{HRT-SOC}}$ which is expressed as follows: 
\begin{equation}
\begin{aligned}
\label{eqn:estimator}
T_{\textrm{HRT-SOC}}= g\left(S_{N}\right)
\prod_{i=1}^N \frac{e^{ -\mu_i(S_{i-1}) \Lambda_{X_i}\left(X_{i}\right)}}{ (1-\mu_i(S_{i-1}))},
\end{aligned}
\end{equation}
where $g(x)=\mathbbm{1}_{(x \leq \gamma_{th})}$ in the first application, and $g(x)=F_{X_{0}}(\gamma_{{th}}(x+\eta))$ in the second application.
In this setting, each step of the backward algorithm can be expressed, for $k=0,\cdots,K$, as 
\begin{equation}
\begin{aligned}
\label{eqn:eqf}
  u(n,s_k)= \min_{\mu \in (-\infty,1)} \; \; \;  \frac{1}{1-\mu} \;  \int_0^{+ \infty} u(n+1,s_k+t) f_{X_{n+1}}(t) e^{-\mu \; \Lambda_{X_{n+1}}(t)} dt.
  \end{aligned}
\end{equation}
The controls $\mu_n(s_k)$ are obtained by solving the following equation:
\begin{equation}
\begin{aligned}
&1-\mu_n (s_k)=
\frac{\int_0^{+ \infty} u(n-1,s_k+t) f_{X_{n-1}}(t) \; e^{-\mu_n(s_k) \; \Lambda_{X_{n-1}}(t)} dt}{\int_0^{+ \infty} \Lambda_{X_{n-1}}(t) u(n-1,s_k+t) f_{X_{n-1}}(t) \; e^{-\mu_n(s_k) \; \Lambda_{X_{n-1}}(t)} dt}.
\end{aligned}
\end{equation}

For the forward step, assuming that the control is smooth (motivated by numerical observations), we can compute the controls by interpolating between the points $\mu_n(s_k),\; k=0,\cdots,K$. \\
To sample from the change of measures $\Tilde{f}_{ X_{i}}
(\cdot), \quad i=1.\cdots, N$, we used the inverse CDF technique. In \cite{rached2015fast}, the authors revealed that the inverse CDF of the HRT of a log-normal RV $X_i$ is given by
\begin{equation}
F_{X_i}^{-1}(y)=\exp \left(m+\sigma \Phi^{-1}\left(1-(1-y)^{-\frac{1}{\mu_i-1}}\right)\right),
\end{equation}
where $\mu_i$ is the twisting parameter corresponding to $X_i$ and $\Phi(\cdot)$ is the CDF of the standard normal distribution. This formula can be generalized to other distributions as \cite[eq.~(65)]{rached2015unified}.


The relative error serves as a measure of the efficiency of the estimators. The relative error of the naive MC estimator and the proposed IS estimator are defined respectively through the central limit theorem~\cite{asmussen2007stochastic} as
 \begin{equation}
     \label{eqn:eq10}
\epsilon_{\textrm{MC}}=C \frac{\sqrt{\alpha(1-\alpha)}}{\sqrt{M} \alpha}, \; \;  \epsilon_{\textrm{HRT-SOC}}=C \frac{\sqrt{\operatorname{Var}\left[T_{\textrm{HRT-SOC}}\right]}}{\sqrt{M} \alpha},
\end{equation}
where $C$ is the confidence constant equal to $1.96$ for the $95$\%
confidence interval.

We compared the estimator defined in (\ref{eqn:estimator}) to other existing estimators when calculating the OP at the output of diversity receivers with and without co-channel interference. For instance, using the log-normal setting with the HRT technique allows us to compare the estimator with the approach in \cite{rached2015fast}, which used the HRT without SOC (i.e., the control is constant, independent of the state and time). We denote this method as HRT. In the numerical experiments, the HRT-SOC technique reduces the variance substantially compared to other approaches. However, it requires additional time, called a backward cost, to determine the optimal controls. \\  We let $ M_{\textrm{HRT}}$ and $ M_{\textrm{HRT-SOC}}$ be the number of required simulation runs for the HRT estimator $T_{\textrm{HRT}}$ and the proposed estimator $T_{\textrm{HRT-SOC}}$, respectively, to ensure a relative error equal to $\textrm{TOL}$. The total costs of the HRT-SOC and HRT approaches are expressed as follows:
 \begin{equation}
     \label{eqn:eq20}
     \text{W}_{\textrm{HRT-SOC}}= \underbrace{N \times K \times T_b}_{\text{Backward cost}} + \underbrace{M_{\textrm{HRT-SOC}} \times T_f}_{\text{Forward cost}},
\end{equation}
 \begin{equation}
     \label{eqn:eqtime}
     \text{W}_{\textrm{HRT}}= \underbrace{M_{\textrm{HRT}} \times T_f}_{\text{Forward cost}},
\end{equation}
where $T_b$ is the time required in the backward algorithm to calculate a single control, and $T_f$ represents the cost per sample in the forward step (approximately the same for both approaches). Figures~\ref{fig10} and \ref{fig12} illustrate that the variance reduction compared to the HRT technique increases as the quantity of interest becomes rarer. Thus, we determine that $M_{\textrm{HRT}} \gg M_{\textrm{HRT-SOC}}$, especially for rare regions. Consequently, we expect that, for the regime of rare events and a fixed $N$, the backward time can be neglected compared to the forward cost of the HRT, which is presented in Figure~\ref{fig11}. \\
When the backward time dominates the forward time of the HRT, we propose an improved version of the HRT-SOC estimator. We call this version the aggregate method (HRT-SOC-AG), which aims to reduce the backward cost without considerably affecting the variance reduction. 

\subsection{Aggregate Method}
The idea for the aggregate method is to divide the sum $S_N$ into $B$ blocks and compute the controls for each block rather than for each $X_i,  i=1,\cdots,N$. Doing so reduces the backward cost from $N \times K \times T_b$ to $B \times K \times T_b$. In other words, if we select $B$ blocks, such that $B \leq N$, we consider the following dynamics: 
 \begin{equation}
     \label{eqn:eqd2}
 S_{n_m+b_{m+1}}=S_{n_m}+ \sum_{i=n_m+1}^{n_m+b_{m+1}} X_{i},\; \;  m=0,1, \cdots, B-1,
\end{equation}
where $n_m=\sum_{j=1}^{m} b_j$, and $b_m, \; \;  m=1,2, \cdots, B$ are chosen such that $n_B=\sum_{j=1}^{B} b_j =N$. 
We adopted the same control $\mu_m(S_{n_{m-1}})$ for each $X_i$ from $i=n_{m-1}+1$ to $i=n_m$. Thus, the $B$ new controls $\mu^X_1,\cdots, \mu^X_B \;$ are defined such that  \begin{equation}
     \label{eqn:eqmu}
{\mu_i}=\mu^X_m\; \; \text{for} \; \; n_{m-1} < i \leq n_m, \;i=1, \cdots, N, \;  m=1, \cdots, B.
 \end{equation}
With this proposed approach, we decreased the cost of the backward step with the price of increasing the variance.\\
To determine $\mu^X_1,\cdots, \mu^X_B$, we used the dynamics proposed in (\ref{eqn:eqd2}) instead of the initial dynamics (\ref{eqn:eqd1}) to define a reformulated dynamic programming equation. We employed the same steps as those followed in the proof of the proposition, but instead of conditioning on $X_{n+1}$, we conditioned on $X_{n_m+1},\cdots, X_{n_m+b_m}$. Applying the same control $\mu_{m+1}$ for each $X_i, i=n_m +1,\cdots,n_m+b_m $, as explained in (\ref{eqn:eqmu}), we obtain 
\begin{equation}
\label{eqn:eq21}
\begin{aligned}
  u(m,s_k)= 
  \min _{\mu \in (-\infty,1)} \int_{[0,+ \infty[^{b_m}} & \frac{ e^{- \mu \sum_{j=n_m+1}^{n_m+b_m} \Lambda_{X_j(t_j)}}}{(1-\mu)^{b_m}}\prod_{j=n_m+1}^{n_m+b_m} f_{X_j}(t_j) \\ 
  &  \times  \; u\left(m+1,s_{k}+t_{n_m+1}+ \cdots+ t_{n_m+b_m} \right) \; d t_{n_m+1} \cdots d t_{n_m+b_m}.
  \end{aligned}
\end{equation}
Instead of solving the above equation, we propose minimizing its approximate upper bound, which becomes clearer in the next two subsections.
\subsection{OP at the Output of Diversity Receivers in a Log-normal Environment without Co-channel Interference} 
The computation of the OP at the output of diversity receivers is equivalent to evaluating the CDF of the sum of the SNRs. Therefore, the interest in the first application is in the estimation of the left-tail region of the following form:
\begin{equation}
\label{eqn:app1}
\mathbb{P}\left(\sum_{i=1}^{N} X_i \leq \gamma_{\textrm{th}} \right).
\end{equation}
 We compared the approach to the HRT technique (see \cite{rached2015fast}) and the exponential twisting estimator (see \cite{asmussen2016exponential}). We also used the improved version to achieve better results. When applying the aggregate method, instead of solving (\ref{eqn:eq21}), we propose to minimize an approximate upper bound of it.
 More precisely, for the i.i.d. log-normal case, 
\begin{equation}
\label{eqn:eqre}
\sum_{j=n_m+1}^{n_m+b_m} \Lambda_{X_j}(t_j) \leq \Lambda_{X}\left( \sum_{j=n_m+1}^{n_m+b_m} t_j\right),  \; t_j>0.
\end{equation} 
holds asymptotically, i.e. when the sum $\sum_{j=n_m+1}^{n_m+b_m} t_j$ is sufficiently small, where $X$ has the same distribution as $X_j, j=n_m+1,\cdots, n_m+b_m$. This result can be proven using the asymptotic result of the tail of a Normal distribution in \cite{asmussen2011efficient}. Using the inequality (\ref{eqn:eqre}), the twisting parameters $\mu^X_{m+1}$ are then selected as the argmin of the following approximated upper bound
 \begin{equation}
 \begin{aligned}
\label{eqn:equ1}
& u(m,s_k) \lessapprox \min _{\mu \in (-\infty,1)} \int_{[0,S-s_k]} \; \; \frac{ e ^{ - \mu \;  \Lambda_{X}(y)}}{(1-\mu)^{b_m}} \; \; f_{\sum_{j=n_m+1}^{n_m+b_m}X_j}\left (y\right) \; u\left(n_m+b_m,s_{k}+y \right) \; d y,
\end{aligned}
\end{equation} 
 where $f_{\sum_{j=n_m+1}^{n_m+b_m}X_j}\left (\cdot\right)$ is the PDF of  $\sum_{j=n_m+1}^{n_m+b_m}X_j$. Given that the PDF of sums of i.i.d. log-normal RVs is unknown, we suggest approximating it using a univariate log-normal PDF $f_{Y_{m+1}}(\cdot)$, whose parameters are computed using moment matching (see \cite{cobb2012approximating}).
 Finally, we obtain
\begin{equation}
\begin{aligned}
\label{eqn:equ2}
& u(m,s_k) \approx \min _{\mu\in (-\infty,1)} \int_{[0,S-s_k]} \; \; \frac{ e ^{ - \mu \; \Lambda_{X}(y)}}{(1-\mu)^{b_m}} \; \; f_{Y_{m+1}}\left (y\right)  \; u\left(n_m+b_m,s_{k}+y \right) \; d y, \; \; m=0,\cdots,B-1.
\end{aligned}
\end{equation} 
Moreover,  $\mu^X_{1},\cdots,\mu^X_{B}$ are obtained as follows: 
\begin{equation}
\begin{aligned}
\label{eqn:equ2.}
 \mu^X_{m+1}(s_k) \approx \underset{\mu\in (-\infty,1)}{\arg \min}  \int_{[0,S-s_k]} \; \; \frac{ e ^{ - \mu \;  \Lambda_{X}(y)}}{(1-\mu)^{b_m}} \; \; f_{Y_{m+1}}\left (y\right) \; u\left(n_m+b_m,s_{k}+y \right) \; d y, \; \; m=0,\cdots,B-1.
\end{aligned}
\end{equation} 
Figure~\ref{fig10} plots the number of samples, required for the various approaches,
to ensure $\textrm{TOL}=5\%$ as a function of $\gamma_{\textrm{th}}$. The range of $\gamma_{\textrm{th}}$ ensures a range of probabilities between $2 \times 10^{-12}$ and $6 \times 10^{-6}$. For the aggregate method, we selected a constant parameter $b$ (i.e., $b_m=2$ for all $m=1,\cdots, B$ with $B=\frac{N}{2}$). \\
The choice of the parameter $K=20$ is motivated by Figure~\ref{figK}, which plots the variance as a function of $K$. A larger $K$ results in a smaller variance. The backward step is costly when $K$ is large. Further, the variance reduction for $K>20$ is minimal compared to the increased cost of solving the backward problem.
\begin{figure}[H] 
\begin{center}
\includegraphics[scale = 0.45]{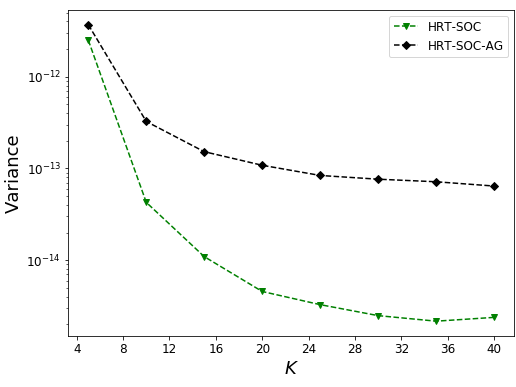} 
\caption {Variance as a function of $K$ with the following parameters: $N=10$, $m=0$~dB,  $\sigma=3$~dB, $\textrm{TOL}=0.05$, and $b=2$.} 
\label{figK}
\end{center}
\end{figure} 
\begin{figure}[H] 
\begin{center}
\includegraphics[scale = 0.45]{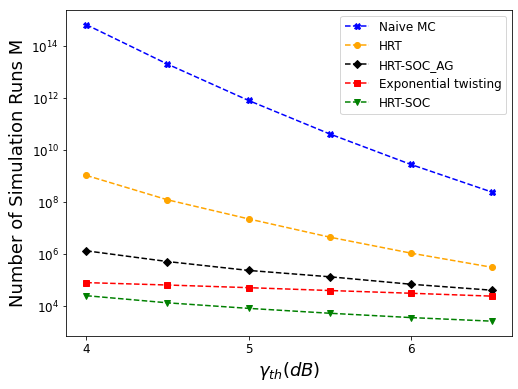} 
\caption {Number of required simulation runs for a $5\%$ relative error with the following parameters: $N=10$, $K=20$, $m=0$~dB, $\sigma=3$~dB, $\textrm{TOL}=0.05$, and $b=2$.} 
\label{fig10}
\end{center}
\end{figure} 
Figure~\ref{fig10} indicates that the number of samples required by naive MC simulations increases faster as the threshold decreases. In addition, the HRT-SOC approach requires the smallest number of simulation runs and saves a considerable number of samples compared to the HRT approach. For example, the number of simulations reduces by about 41,775 times for a small threshold (4~dB), corresponding to an OP value of $2 \times 10^{-12}$. In contrast, the HRT-SOC-AG requires an additional number of samples, compared to the HRT-SOC approach, to reach a $5\%$ relative error, indicating that the variance has increased as expected. However, we still obtained better variance reduction compared to the HRT technique. 

We further studied the computational work for each method. Figure~\ref{fig11} plots the total time required for the exponential twisting, HRT, HRT-SOC, and HRT-SOC-AG techniques to ensure a 5\% relative error as a function of the threshold. We also plotted the time required by the HRT-SOC and HRT-SOC-AG techniques to demonstrate the time required for the backward step compared to that required for the forward step.

The proposed estimator is the best for computational time for small thresholds (corresponding to an OP of less than $3.6 \times 10^{-8}$). As the event becomes rarer, the time gap between the proposed approach and other IS techniques increases significantly. Additionally, Figures~\ref{fig10} and \ref{fig11} reveal that the HRT approach requires numerous samples to estimate the OP of the order of $2 \times 10^{-12}$ with good accuracy. However, for an OP greater than $3.6 \times 10^{-8}$, the proposed approach is more expensive than others due to the additional computational time for the backward step for each threshold, which exceeds the time the remaining techniques when the number of samples is not sufficiently large. Nevertheless, this was enhanced when we used the improved version. The HRT-SOC-AG reduces the CPU time by about 1.7 times compared to the HRT-SOC approach for $\gamma_{\textrm{th}} \geq 5$~dB. Thus, with this choice of $b$, the efficiency of the aggregate method regarding time reduction exceeds the loss in variance. This choice of $b_m, m=1,\cdots, B$ is not optimal. Despite this, it provides better results than the HRT-SOC approach. 
\begin{figure}[H] 
	\begin{center}
		\includegraphics[scale = 0.45]{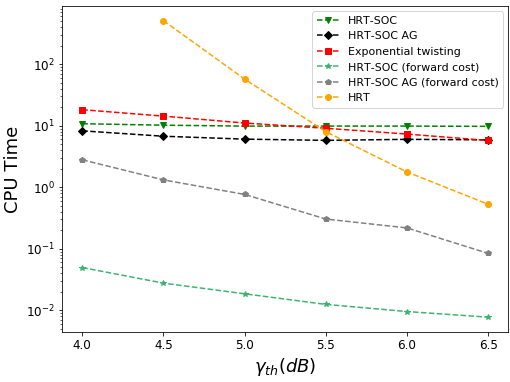} 
		\caption {CPU time required for a 5\% relative error with the following parameters: $N=10$, $K=20$, $m=0$~dB, $\sigma=3$~dB, $\textrm{TOL}=0.05$, and $b=2$.} 
		\label{fig11}
	\end{center}
\end{figure}
Another possible experiment is to study the efficiency as a function of the number $N$ of antennae for a fixed threshold and investigate the number of simulation runs required for each method and the computational time (Figures~\ref{fig12} and \ref{fig13}, respectively). The range of the OP is between $10^{-5}$ and $2.5 \times 10^{-12}$ when using a range between nine and 13 antennae and a fixed threshold $\gamma_{\textrm{th}}=6$~dB. For the aggregate method, we used $b_m=2, \; m=1,\cdots, \frac{N}{2}$ for an even-numbered $N$ and $b_m=2,m=1,\cdots,{\frac{N-3}{2}}, \; b_{\frac{N-1}{2}}=3$ for an odd-numbered $N$.
   \begin{figure}[H] 
\begin{center}
\includegraphics[scale = 0.45]{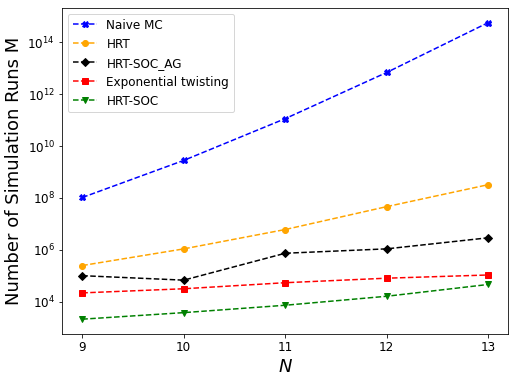} 
\caption {Number of required simulation runs for a 5\% relative error with the following
parameters: $K=20$, $\gamma_{\textrm{th}}=6$~dB, $m=0$~dB, $\sigma= 3$~dB, and $\textrm{TOL}=0.05$.} 
\label{fig12}
\end{center}
\end{figure}
\begin{figure}[H] 
	\begin{center}
		\includegraphics[scale = 0.45]{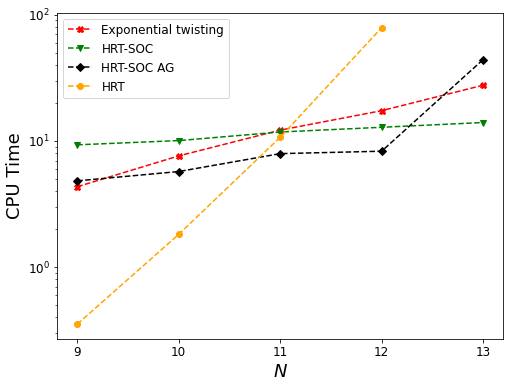}
		\caption {CPU time required for a 5\% relative error with the following parameters: $K=20$, $\gamma_{\textrm{th}}=6$~dB, $m=0$~dB, $\sigma=3$~dB, and $\textrm{TOL}=0.05$.} 
		\label{fig13}
	\end{center}
\end{figure}
Figure~\ref{fig12} indicates that the HRT-SOC approach is more efficient and requires fewer simulation runs than the HRT and the exponential twisting approaches. For $N=13$, the proposed method requires 7,455 times fewer simulation runs than the HRT technique to meet the same accuracy requirements. In addition, the variance reduction for the HRT-SOC-AG technique depends on whether $N$ is odd or even. Moreover, the HRT-SOC-AG method requires more simulation runs than the HRT-SOC technique to reach a fixed precision $\textrm{TOL}$, but it is more efficient in terms of CPU time for $N \leq 12$. When the event becomes rarer (for small $\gamma_{\textrm{th}}$ and large $N$), the improved approach with a fixed choice of $b$ becomes less efficient in terms of CPU time than the HRT-SOC approach. In these cases, the number of samples is large enough that the backward time is neglected. Thus, reducing the variance rather than the cost of the backward step is more efficient. These results demonstrate that the choice of $b_m, \; m=1,\cdots, B$ is crucial and should be adaptively chosen to provide better results. More precisely, for fixed parameters $\gamma_{\textrm{th}}$, $\textrm{TOL}$ and $N$, the following optimization problem should be solved: 
\begin{equation}
\label{optimization}
\min_{b,M,K}  \; \; \; \; B \times K \times T_b + M_{\textrm{HRT-SOC-AG}}(b) \times T_f, 
\end{equation}
such that
$$
C^2 \frac{ \operatorname{Var}\left[T_{\textrm{HRT-SOC-AG}}(b)\right]}{M_{\textrm{HRT-SOC-AG}}(b) \alpha^2} \leq \textrm{TOL}^2.$$
The above optimization problem reveals that an optimal choice of $b_m$ in the case of a very rare event is $b_m=1, \; m=1,\cdots, B$, where $B=N$. However, when the event becomes less rare, an optimal choice of $B$ is to take a single block (i.e., $b_1=N$). By doing so, the HRT-SOC-AG technique reduces to the HRT technique because the controls are state-independent in this case. Future work can be devoted to solving the previous optimization problem. Using optimal values of $b_m$, we expect the HRT-SOC-AG estimator to achieve better performance.

\subsection{OP in the Presence of Co-channel Interference in a Log-normal Environment for SISO Systems}
We consider a SISO system and recall that the OP in the presence of co-channel interference and noise is expressed as follows: 
  \begin{equation*}
      P_{out}=
      \mathbb{E}\left[F_{X_{0}}\left(\gamma_{\textrm{th}}\left(\sum_{n=1}^{N} X_{n}+\eta \right)\right)\right],
\end{equation*}
where $X_1,\cdots,X_N$ are the interfering power signal and are assumed to be i.i.d. log-normal RVs with parameters $m$ and $\sigma^2$.
\begin{figure}[H] 
	\begin{center}
		\includegraphics[scale = 0.55]{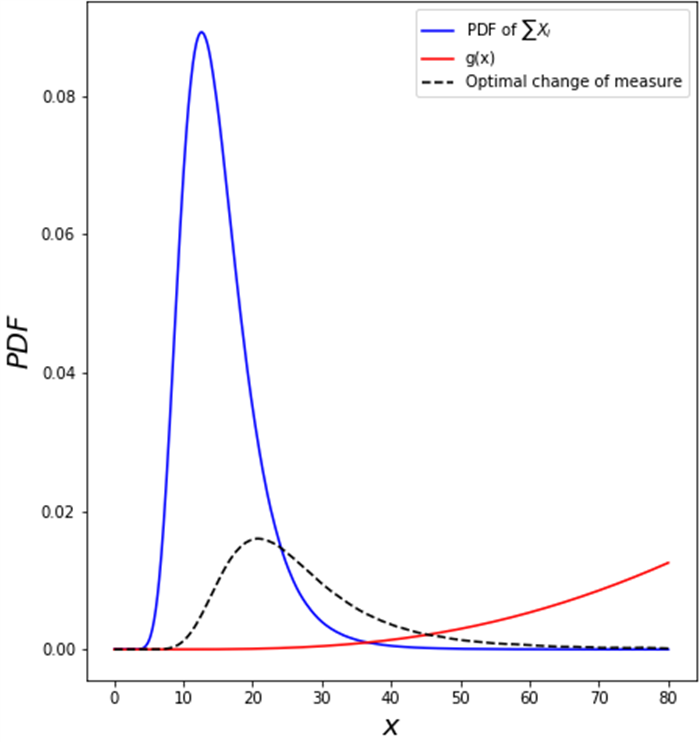}     
		\caption {Motivation for using IS with $N=10$, $m_0=10$~dB, $\sigma_0=4$~dB, $m=0$~dB, $\sigma=4$~dB, $\gamma_{\textrm{th}}=-18$~dB, and $\eta=-10$~dB.}
		\label{figmotiv}
	\end{center}
\end{figure}
The PDF of $\sum_{i=1}^N X_i$ is denoted by $f_{\sum_{i=1}^N X_i}(\cdot)$. To motivate the need for IS to efficiently estimate $P_{out}$, Figure~\ref{figmotiv} plots the quantities $f_{\sum_{i=1}^N X_i}$, $g$, and the optimal IS PDF, which is proportional to $g  f_{\sum_{i=1}^N X_i}$. The product $g  f_{\sum_{i=1}^N X_i}$ in Figure~\ref{figmotiv} is not normalized (i.e., it is an unnormalized PDF).

Sampling from the original PDF of $\sum_{i=1}^N X_i$ is not efficient (i.e., when sampling from the original PDF, most samples fall in the region where $g$ takes almost zero values). Hence, the computation of $P_{out}$ behaves like a rare event problem and can be addressed using the proposed HRT-SOC technique. The comparison is made concerning the estimator of \cite{rached2017efficient}, which is based on a covariance matrix scaling (CS) technique. It transforms the problem of evaluating the OP to computing the probability that a sum of correlated log-normal RVs exceeds a certain threshold. The estimator in \cite{rached2017efficient} is given by 
\begin{equation}
T_{\textrm{CS}}(\mathbf{Z})=\mathbbm{1}_{\left(\sum_{i=0}^{N} \exp \left(Z_{i}\right) \geq 1 / \gamma_{th}\right)} L\left(Z_{0}, \ldots, Z_{N}\right),
\end{equation}
where $\mathbf{Z}=\left(Z_{0}, Z_{1}, \ldots, Z_{N}\right)^{t}$, $\quad Z_{i}= \begin{cases}\log(X_{i})-\log(X_{0}) & i=1,2, \ldots, N \\ \log (\eta)-\log(X_{0})& i=0\end{cases}$, and
\begin{equation}
L\left(Z_{0}, Z_{1}, \ldots, Z_{N}\right)=\frac{\exp \left(-\frac{\theta}{2}(\mathbf{Z}-\boldsymbol{m})^{t} \Sigma^{-\mathbf{1}}(\mathbf{Z}-\boldsymbol{m})\right)}{(1-\theta)^{(N+1) / 2}}.
\end{equation}
The expressions of $\boldsymbol{m}$, $\bold{\Sigma}$, and $\theta$ are given in \cite[eq.~(6)]{rached2017efficient}, \cite[eq.~(7)]{rached2017efficient}, and \cite[eq.~(19)]{rached2017efficient} respectively. \\
We also compared the proposed approach to the exponentially tilted (ET) estimator of \cite{botev2017accurate}.
We also used the HRT-SOC-AG method proposed in the previous subsection to further improve the computational work of the HRT-SOC technique. The reformulated dynamic programming equation is 
\begin{equation}
\label{eqn:eq21intrf}
\begin{aligned}
  u(m,s_k)=  \min _{\mu\in (- \infty, 1)}  \int_{(0,+\infty)^{b_m}}&\frac{ e^{- \mu \sum_{j=n_m+1}^{n_m+b_m} \Lambda_{X_j}(t_j)}}{(1-\mu)^{b_m}} \\& \times \prod_{j=n_m+1}^{n_m+b_m} f_{X_j}(t_j) u\left(m+1,s_{k}+\sum_{j=n_m+1}^{n_m+b_m} t_j \right) \; d t_{n_m+1} \cdots d t_{n_m+b_m}.
  \end{aligned}
\end{equation}
\\
Next, using the following inequality, proven in \cite{juneja2002simulating}, which is particularly satisfied in the case of i.i.d. log-normal RVs and holds for $\sum_{j=n_m+1}^{n_m+b_m} t_j$ that are large enough:
\begin{equation}
\label{eqn:eqreintrf}
\sum_{j=n_m+1}^{n_m+b_m} \Lambda_{X_j}(t_j) \geq \Lambda_{X}\left( \sum_{j=n_m+1}^{n_m+b_m} t_j\right)-\epsilon,  \; t_j>0, \; \text{for all}  \; \epsilon>0,
\end{equation} 
we can write 
\begin{equation}
\begin{aligned}
\label{eqn:equ3}
u(m,s_k) \approx & \min _{\mu \in (- \infty, 1)} \int_{(0,+\infty)} \; \; \frac{ e ^{ - \mu \; \Lambda_{X}(y)}}{(1-\mu)^{b_m}} \; \; f_{Y_{m+1}}\left (y\right) \; u\left(n_m+b_m,s_{k}+y \right) \; d y, \; \; m=0,\cdots, B-1.
\end{aligned}
\end{equation} 
The large value of $\sum_{j=n_m+1}^{n_m+b_m} t_j$ is motivated by Figure~\ref{figmotiv}, which illustrates that the change of measure tends to increase the value of the sum in the regime of rare events.
We studied the efficiency of the four IS schemes regarding the number of samples necessary to ensure a fixed accuracy requirement. To this end, Figure~\ref{figintrf2} plots the number of samples to ensure $\textrm{TOL}=5\%$ as a function of $\gamma_{\textrm{th}}$. This figure reveals that the HRT-SOC approach saves numerous samples compared to other approaches. For instance, the CS technique requires approximately 2,000 times as many simulations as the HRT-SOC scheme needs. The aggregate method did not affect the variance reduction. 

We further investigated the gain in terms of the required computational time. Figure~\ref{figintrf3} presents the total CPU time needed by the four techniques to achieve the fixed accuracy TOL. The HRT-SOC approach requires less CPU time than the ET approach for the range of considered thresholds. In particular, when $\gamma_{\textrm{th}}=-30$~dB, it is 13 times more efficient than the ET scheme. Compared to the CS approach, the HRT-SOC technique is more efficient when  $\gamma_{\textrm{th}}<-25$~dB, corresponding to an OP less than $3 \times 10^{-8}$. The required computational time for the HRT-SOC technique is almost the same in the considered threshold range, whereas the CS and ET approaches require much more time as the threshold decreases. Moreover, the HRT-SOC-AG technique requires less time than the HRT-SOC technique using $b=2$ to estimate the quantity of interest $\alpha$. Therefore, the improved approach widens the region over which the proposed approach outperforms the CS approach.
\begin{figure}[H] 
\begin{center}
\includegraphics[scale = 0.45]{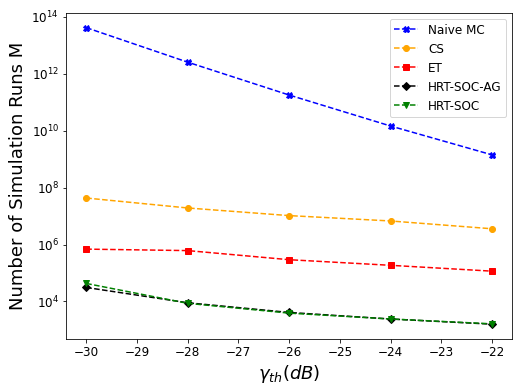}   
\caption {Number of required simulation runs with the following parameters: $N = 10$, $K=20$, $S=40$, $\textrm{TOL}=0.05$, $\eta=-10$~dB, $ m_0=10$~dB, $\sigma_0=4$~dB, $m=0$~dB, and $\sigma=4$~dB.}
\label{figintrf2}
\end{center}
\end{figure}
\begin{figure}[H] 
\begin{center}
\includegraphics[scale = 0.45]{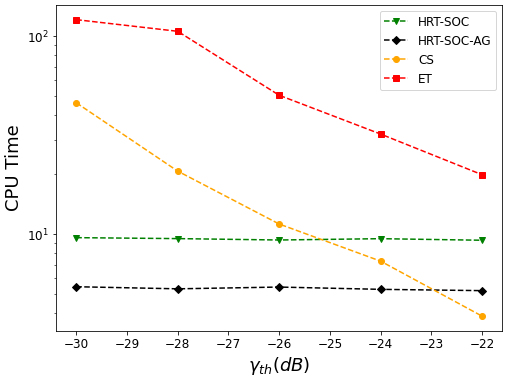}     
\caption {CPU time required for a 5\% relative error with the following parameters: $N = 10$, $K=20$, $S=40$, $\textrm{TOL}=0.05$, $\eta=-10$~dB,  $m_0=10$~dB, $\sigma_0=4$~dB,  $m=0$~dB, and $\sigma=4$~dB.} 
\label{figintrf3}
\end{center}
\end{figure}
\newpage
In the last experiment, we studied the influence of varying the accuracy $\textrm{TOL}$ on the proposed and other IS approaches. To this end, Figures~\ref{figintrf7} and \ref{figintrf8} present the number of simulation runs and CPU time needed when varying $\textrm{TOL}$ for a fixed $\gamma_{\textrm{th}}$ and $N$. This choice makes the OP approximately equal to $10^{-7}$. 
\begin{figure}[H] 
\begin{center}
\includegraphics[scale = 0.45]{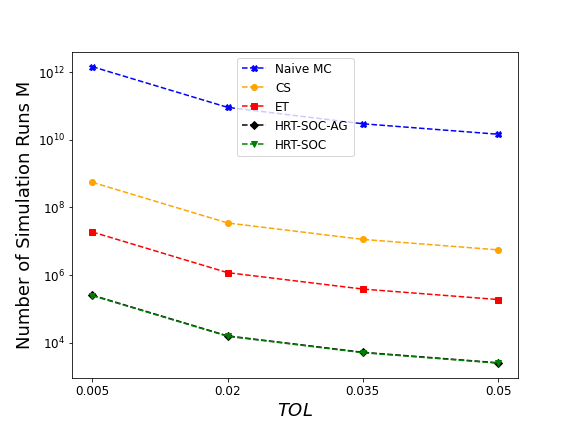}     
\caption {Number of required simulation runs with the following parameters: $N = 10$, $K=20$, $S=40$, $\gamma_{\textrm{th}}=-24$~dB, $\eta=-10$~dB, $m_0=10$~dB, $\sigma_0=4$~dB, $m=0$~dB, and $\sigma=4$~dB.}
\label{figintrf7}
\end{center}
\end{figure}
\begin{figure}[H] 
\begin{center}
\includegraphics[scale = 0.45]{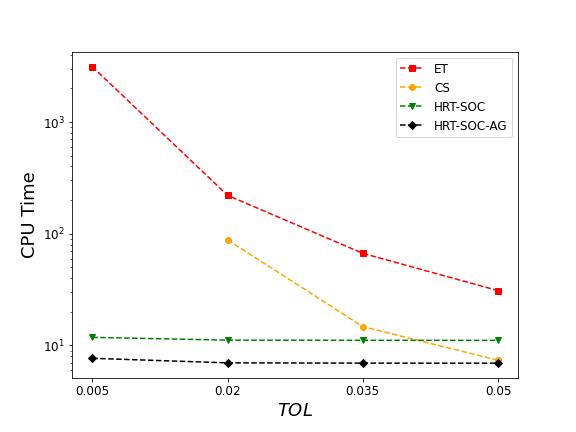}     
\caption {CPU time for a 5\% relative error with the following parameters: $N=10$, $K=20$, $S=40$, $\gamma_{\textrm{th}}=-24$~dB, $\eta=-10$~dB,  $m_0=10$~dB, $\sigma_0=4$~dB, $m=0$~dB, and $\sigma=4$~dB.}
\label{figintrf8}
\end{center}
\end{figure}
Figure~\ref{figintrf7} confirms the high gains of the proposed methods compared to all other IS approaches. Our approaches are $2000$ times (respectively $65$ times) more efficient than the CS (respectively the ET) approaches for all values of $\textrm{TOL}$. Furthermore, Figure~\ref{figintrf8} demonstrates that the required time for the proposed methods compared to the other algorithms remains unchanged for the considered range of $\textrm{TOL}$. Moreover, similarly to the previous conclusions, the computational time required by the proposed algorithm is less than that needed by the ET algorithm for all $\textrm{TOL}$. Additionally, the superior performance of the method compared to the CS approach is critical for small values of TOL. Finally, the HRT-SOC-AG method increases the threshold, below which the proposed method performs better than the CS approach, from 0.045 to 0.058.
\vspace{-1mm}

\section{Conclusions}
\subsection*{Summary}
We developed a generic state-dependent IS algorithm to efficiently estimate rare event quantities that could be written as an expectation of a functional of the sums of independent RVs. These problems have applications in the performance analysis of wireless communications systems operating over fading channels. Within a preselected class of a change of measures, the optimal IS parameters are determined via the connection to an SOC formulation. The numerical experiments verified the ability of the proposed approach to accurately and efficiently estimate the quantity of interest in the rare event regime. The proposed approach yields a substantial variance reduction compared with other well-known estimators. Additionally, the estimator requires less CPU time than the other proposed approaches in rare regions. We also proposed an aggregate method to improve efficiency further in terms of computational time.

\subsection*{Possible Extensions}
For future research, the present work can be extended in many directions. One possible direction is to optimize the aggregate method by solving the optimization problem (\ref{optimization}). 
 \\
A further interesting extension is to consider multivariate RVs when estimating the quantity of interest. In this case, RVs should be mutually independent, but the components of each RV are not necessarily independent. As the backward cost increases exponentially with the dimensions, we could employ cheaper approximation methods to calculate the controls, such as neural
networks.\\
\textbf{Acknowledgments} This publication is based on work supported by the King Abdullah University of Science and Technology (KAUST) Office of Sponsored Research (OSR) under Award No. OSR-2019-CRG8-4033 and the Alexander von Humboldt Foundation.
\bibliographystyle{plain}
\bibliography{References.bib}

\begin{thebibliography}{10}

\bibitem{alemany2013nonparametric}
Ramon Alemany, Catalina Bolanc{\'e}, and Montserrat Guillen.
\newblock A nonparametric approach to calculating value-at-risk.
\newblock {\em Insurance: Mathematics and Economics}, 52(2):255--262, 2013.

\bibitem{am2012state}
Karthyek~Rajhaa AM and Sandeep Juneja.
\newblock State-independent importance sampling for estimating large deviation
  probabilities in heavy-tailed random walks.
\newblock In {\em proc. of the IEEE 6th International ICST Conference on
  Performance Evaluation Methodologies and Tools}, pages 127--135, 2012.

\bibitem{asmussen1997simulation}
S{\o}ren Asmussen and Klemens Binswanger.
\newblock Simulation of ruin probabilities for subexponential claims.
\newblock {\em ASTIN Bulletin: The Journal of the IAA}, 27(2):297--318, 1997.

\bibitem{asmussen2000rare}
S{\o}ren Asmussen, Klemens Binswanger, Bjarne H{\o}jgaard, et~al.
\newblock Rare events simulation for heavy-tailed distributions.
\newblock {\em Bernoulli}, 6(2):303--322, 2000.

\bibitem{asmussen2011efficient}
S{\o}ren Asmussen, Jos{\'e} Blanchet, Sandeep Juneja, and Leonardo
  Rojas-Nandayapa.
\newblock Efficient simulation of tail probabilities of sums of correlated
  {L}ognormals.
\newblock {\em Annals of Operations Research}, 189(1):5--23, 2011.

\bibitem{asmussen2007stochastic}
S{\o}ren Asmussen and Peter~W Glynn.
\newblock {\em Stochastic simulation: algorithms and analysis}.
\newblock Springer Science \& Business Media, 2007.

\bibitem{asmussen2016exponential}
S{\o}ren Asmussen, Jens~Ledet Jensen, and Leonardo Rojas-Nandayapa.
\newblock Exponential family techniques for the {L}ognormal left tail.
\newblock {\em Scandinavian Journal of Statistics}, 43(3):774--787, 2016.

\bibitem{asmussen2012error}
S{\o}ren Asmussen and Dominik Kortschak.
\newblock On error rates in rare event simulation with heavy tails.
\newblock In {\em proc. of the IEEE Winter Simulation Conference (WSC)}, pages
  1--11, 2012.

\bibitem{asmussen2015error}
S{\o}ren Asmussen and Dominik Kortschak.
\newblock Error rates and improved algorithms for rare event simulation with
  heavy weibull tails.
\newblock {\em Methodology and Computing in Applied Probability},
  17(2):441--461, 2015.

\bibitem{asmussen2006improved}
S{\o}ren Asmussen and Dirk~P Kroese.
\newblock Improved algorithms for rare event simulation with heavy tails.
\newblock {\em Advances in Applied Probability}, 38(2):545--558, 2006.

\bibitem{bassamboo2007inefficiency}
Achal Bassamboo, Sandeep Juneja, and Assaf Zeevi.
\newblock On the inefficiency of state-independent importance sampling in the
  presence of heavy tails.
\newblock {\em Operations research letters}, 35(2):251--260, 2007.

\bibitem{beaulieu2020marcum}
Norman~C Beaulieu and Gan Luan.
\newblock On the {M}arcum {Q}-function behavior of the left tail probability of
  the {L}ognormal sum distribution.
\newblock In {\em proc. of the IEEE International Conference on Communications
  (ICC)}, 2020.

\bibitem{issaid2017efficient}
Chaouki Ben~Issaid, Nadhir Ben~Rached, Abla Kammoun, Mohamed-Slim Alouini, and
  Raul Tempone.
\newblock On the efficient simulation of the distribution of the sum of
  {G}amma--{G}amma variates with application to the outage probability
  evaluation over fading channels.
\newblock {\em IEEE Transactions on Communications}, 65(4):1839--1848, 2017.

\bibitem{rached2015fast}
Nadhir Ben~Rached, Fatma Benkhelifa, Mohamed-Slim Alouini, and Raul Tempone.
\newblock A fast simulation method for the {L}og-normal sum distribution using
  a hazard rate twisting technique.
\newblock In {\em proc. of the IEEE International Conference on Communications
  (ICC)}, pages 4259--4264, 2015.

\bibitem{rached2018generalization}
Nadhir Ben~Rached, Fatma Benkhelifa, Abla Kammoun, Mohamed-Slim Alouini, and
  Raul Tempone.
\newblock On the generalization of the hazard rate twisting-based simulation
  approach.
\newblock {\em Statistics and Computing}, 28(1):61--75, 2018.

\bibitem{large_sum}
Nadhir Ben~Rached, Abdul-Lateef Haji-Ali, Gerardo Rubino, and Raúl Tempone.
\newblock Efficient importance sampling for large sums of independent and
  identically distributed random variables.
\newblock {\em Statistics and Computing}, 31(6):4353--4362, 2021.

\bibitem{rached2015unified}
Nadhir Ben~Rached, Abla Kammoun, Mohamed-Slim Alouini, and Raul Tempone.
\newblock Unified importance sampling schemes for efficient simulation of
  outage capacity over generalized fading channels.
\newblock {\em IEEE Journal of Selected Topics in Signal Processing},
  10(2):376--388, 2015.

\bibitem{rached2017efficient}
Nadhir Ben~Rached, Abla Kammoun, Mohamed-Slim Alouini, and Raul Tempone.
\newblock On the efficient simulation of outage probability in a {L}og-normal
  fading environment.
\newblock {\em IEEE Transactions on Communications}, 65(6):2583--2593, 2017.

\bibitem{rached2020accurate}
Nadhir Ben~Rached, Abla Kammoun, Mohamed-Slim Alouini, and Ra{\'u}l Tempone.
\newblock An accurate sample rejection estimator of the outage probability with
  equal gain combining.
\newblock {\em IEEE Open Journal of the Communications Society}, 1:1022--1034,
  2020.

\bibitem{rached2020universal}
Nadhir Ben~Rached, Daniel MacKinlay, Zdravko Botev, Ra{\'u}l Tempone, and
  Mohamed-Slim Alouini.
\newblock A universal splitting estimator for the performance evaluation of
  wireless communications systems.
\newblock {\em IEEE Transactions on Wireless Communications}, 19(7):4353--4362,
  2020.

\bibitem{blanchet2012state}
Jose Blanchet and Henry Lam.
\newblock State-dependent importance sampling for rare-event simulation: An
  overview and recent advances.
\newblock {\em Surveys in Operations Research and Management Science},
  17(1):38--59, 2012.

\bibitem{blanchet2011rare}
Jose Blanchet and Chenxin Li.
\newblock Efficient rare event simulation for heavy-tailed compound sums.
\newblock {\em ACM Transactions on Modeling and Computer Simulation (TOMACS)},
  21(2):1--23, 2011.

\bibitem{blanchet2006efficient}
Jose~H Blanchet and Jingchen Liu.
\newblock Efficient simulation for large deviation probabilities of sums of
  heavy-tailed increments.
\newblock In {\em proc. of the IEEE Winter Simulation Conference}, pages
  757--764, 2006.

\bibitem{blanchet2008state}
Jose~H Blanchet and Jingchen Liu.
\newblock State-dependent importance sampling for regularly varying random
  walks.
\newblock {\em Advances in Applied Probability}, 40(4):1104--1128, 2008.

\bibitem{botev2017accurate}
Zdravko Botev and Pierre l'Ecuyer.
\newblock Accurate computation of the right tail of the sum of dependent
  {L}og-normal variates.
\newblock In {\em 2017 Winter Simulation Conference (WSC)}, pages 1880--1890.
  IEEE, 2017.

\bibitem{chan2011rare}
Joshua~CC Chan and Dirk~P Kroese.
\newblock Rare-event probability estimation with conditional {M}onte {C}arlo.
\newblock {\em Annals of Operations Research}, 189(1):43--61, 2011.

\bibitem{chatterjee2018downlink}
Aritra Chatterjee, Priyabrata Parida, and Suvra~Sekhar Das.
\newblock Downlink signal-to-interference ratio and spectral efficiency of
  {MIMO} cellular networks using truncated {L}ognormal approximation.
\newblock {\em IEEE Systems Journal}, 13(1):76--87, 2018.

\bibitem{cobb2012approximating}
Barry~R Cobb, Rafael Rum{\'\i}, and Antonio Salmer{\'o}n.
\newblock Approximating the distribution of a sum of {L}og-normal random
  variables.
\newblock {\em Statistics and Computing}, 16(3):293--308, 2012.

\bibitem{constantinescu2016closed}
Corina Constantinescu, Gennady Samorodnitsky, and Wei Zhu.
\newblock Closed-form ruin probabilities in classical risk models with gamma
  claims.
\newblock 2016.

\bibitem{dupuis2007importance}
Paul Dupuis, Kevin Leder, and Hui Wang.
\newblock Importance sampling for sums of random variables with regularly
  varying tails.
\newblock {\em ACM Transactions on Modeling and Computer Simulation (TOMACS)},
  17(3):14--es, 2007.

\bibitem{dupuis2004importance}
Paul Dupuis and Hui Wang.
\newblock Importance sampling, large deviations, and differential games.
\newblock {\em Stochastics: An International Journal of Probability and
  Stochastic Processes}, 76(6):481--508, 2004.

\bibitem{dupuis2007subsolutions}
Paul Dupuis and Hui Wang.
\newblock Subsolutions of an isaacs equation and efficient schemes for
  importance sampling.
\newblock {\em Mathematics of Operations Research}, 32(3):723--757, 2007.

\bibitem{dupuis2005dynamic}
Paul Dupuis, Hui Wang, et~al.
\newblock Dynamic importance sampling for uniformly recurrent markov chains.
\newblock {\em Annals of Applied Probability}, 15(1A):1--38, 2005.

\bibitem{ghamami2012improving}
Samim Ghamami and Sheldon~M Ross.
\newblock Improving the {A}smussen--{K}roese-type simulation estimators.
\newblock {\em Journal of Applied Probability}, 49(4):1188--1193, 2012.

\bibitem{glasserman2004monte}
Paul Glasserman.
\newblock {\em Monte Carlo methods in financial engineering}, volume~53.
\newblock Springer, 2004.

\bibitem{hammouda2021optimal}
Chiheb~Ben Hammouda, Nadhir~Ben Rached, Ra{\'u}l Tempone, and Sophia Wiechert.
\newblock Optimal importance sampling via stochastic optimal control for
  stochastic reaction networks.
\newblock {\em arXiv preprint arXiv:2110.14335}, 2021.

\bibitem{hartinger2009efficiency}
J{\"u}rgen Hartinger and Dominik Kortschak.
\newblock On the efficiency of the {A}smussen--{K}roese-estimator and its
  application to stop-loss transforms.
\newblock {\em Bl{\"a}tter der DGVFM}, 30(2):363, 2009.

\bibitem{hartmann2017variational}
Carsten Hartmann, Lorenz Richter, Christof Sch{\"u}tte, and Wei Zhang.
\newblock Variational characterization of free energy: Theory and algorithms.
\newblock {\em Entropy}, 19(11):626, 2017.

\bibitem{juneja2007estimating}
Sandeep Juneja.
\newblock Estimating tail probabilities of heavy tailed distributions with
  asymptotically zero relative error.
\newblock {\em Queueing Systems}, 57(2):115--127, 2007.

\bibitem{juneja2002simulating}
Sandeep Juneja and Perwez Shahabuddin.
\newblock Simulating heavy tailed processes using delayed hazard rate twisting.
\newblock {\em ACM Transactions on Modeling and Computer Simulation (TOMACS)},
  12(2):94--118, 2002.

\bibitem{juneja2006rare}
Sandeep Juneja and Perwez Shahabuddin.
\newblock Rare-event simulation techniques: an introduction and recent
  advances.
\newblock {\em Handbooks in operations research and management science},
  13:291--350, 2006.

\bibitem{kroese2013handbook}
D.~P. Kroese, T.~Taimre, and Z.I. Botev.
\newblock {\em Handbook of {M}onte {C}arlo methods}.
\newblock Wiley, N.J, 2011.

\bibitem{lopez2009simple}
Jos{\'e}~A L{\'o}pez-Salcedo.
\newblock Simple closed-form approximation to {R}icean sum distributions.
\newblock {\em IEEE Signal Processing Letters}, 16(3):153--155, 2009.

\bibitem{murthy2015state}
Karthyek~RA Murthy, Sandeep Juneja, and Jose Blanchet.
\newblock State-independent importance sampling for random walks with regularly
  varying increments.
\newblock {\em Stochastic Systems}, 4(2):321--374, 2015.

\bibitem{singh2018exact}
Rajeev~Kumar Singh.
\newblock Exact analytical ber expression in {FSO} using $\alpha$-$\mu$
  distribution proposed as an approximation of {G}amma-{G}amma pdf.
\newblock In {\em proc. of the 27th Wireless and Optical Communication
  Conference (WOCC)}, 2018.

\bibitem{sun2009general}
Lihua Sun and L~Jeff Hong.
\newblock A general framework of importance sampling for value-at-risk and
  conditional value-at-risk.
\newblock In {\em Proceedings of the 2009 Winter Simulation Conference (WSC)},
  pages 415--422. IEEE, 2009.

\bibitem{xiao2019outage}
Zhiqiang Xiao, Bingcheng Zhu, Julian Cheng, and Yongjin Wang.
\newblock Outage probability bounds of {EGC} over dual-branch non-identically
  distributed independent {L}ognormal fading channels with optimized
  parameters.
\newblock {\em IEEE Transactions on Vehicular Technology}, 68(8):8232--8237,
  2019.

\bibitem{yilmaz2012unified}
Ferkan Yilmaz and Mohamed-Slim Alouini.
\newblock A unified {MGF}-based capacity analysis of diversity combiners over
  generalized fading channels.
\newblock {\em IEEE Transactions on Communications}, 60(3):862--875, 2012.

\bibitem{zhu2020right}
Bingcheng Zhu, Zaichen Zhang, Lei Wang, Jian Dang, Liang Wu, Julian Cheng, and
  Geoffrey~Ye Li.
\newblock Right tail approximation for the distribution of {L}ognormal sum and
  its applications.
\newblock In {\em proc. of the IEEE Globecom Workshops}, 2020.

\end{thebibliography}
\end{document}